\documentclass[11pt]{article}
\usepackage{fullpage}
\usepackage{amsfonts, amsmath, amssymb, amsthm}

\newcommand {\ONE}      {\text{\textbf{1}}}

\newcommand {\Exp}       {\mathbb{E}}
\newcommand {\Prob}  [1] {\Pr \Big(#1 \Big)}
\newcommand {\E}     [1] {\Exp\Big[#1\Big]}

\theoremstyle{plain}
\newtheorem{theorem}{Theorem}
\newtheorem{corollary}[theorem]{Corollary}
\newtheorem{lemma}[theorem]{Lemma}

\theoremstyle{remark}
\newtheorem{remark}[theorem]{Remark}

\theoremstyle{definition}
\newtheorem{definition}[theorem]{Definition}

\def\o1{OPT_1}

\begin{document}
\title{On Parsimonious Explanations for\\ 2-D Tree- and Linearly-Ordered Data}
\author 
{
Howard Karloff\thanks
{AT\&T Labs -- Research, 180 Park Avenue, Florham Park, NJ 07932. 
E-mail: {\tt howard@research.att.com}.}
\and
Flip Korn\thanks
{AT\&T Labs -- Research, 180 Park Avenue, Florham Park, NJ 07932. 
E-mail: {\tt flip@research.att.com}.}
\and Konstantin Makarychev\thanks
{IBM Research, Box 218, Yorktown Heights, NY 10598.  
E-mail: {\tt konstantin@us.ibm.com}.}
\and
Yuval Rabani\thanks
{The Rachel and Selim Benin School of Computer 
Science and Engineering, The Hebrew University of Jerusalem,
Jerusalem 91904, Israel.
E-mail: {\tt yrabani@cs.huji.ac.il}.}
}

\date{$ $}
\maketitle
\begin{abstract}

This paper studies the ``explanation problem''
for tree- and linearly-ordered array data,
a problem motivated by database applications
and recently solved for the one-dimensional tree-ordered case.
In this paper, one is given a matrix $A=(a_{ij})$ whose
rows and columns have semantics:
special subsets of the rows and special subsets of the columns
are meaningful, others are not.
A submatrix in $A$ is said to be meaningful if and only if
it is the cross product of a meaningful row subset and a
meaningful column subset, in which case we call it an ``allowed rectangle.''
The goal is to ``explain'' $A$ as a sparse sum of
weighted allowed rectangles.
Specifically, we wish to find as few weighted allowed rectangles
as possible such that, for all $i,j$,
$a_{ij}$ equals the sum of the weights of all rectangles
which include cell $(i,j)$.

In this paper we consider the natural cases in which
the matrix dimensions are tree-ordered or linearly-ordered.
In the tree-ordered case, we are given a rooted tree $T_1$
whose leaves are the rows of $A$ and another, $T_2$,
whose leaves are the columns.  Nodes of the trees
correspond in an obvious way to the sets of their leaf descendants.
In the linearly-ordered case, a set of rows or columns is meaningful
if and only if it is contiguous.

For tree-ordered data, we prove the explanation problem NP-Hard 
and give a randomized $2$-approximation algorithm for it.
For linearly-ordered data, we prove the explanation problem NP-Hard
and give a $2.56$-approximation algorithm.
To our knowledge, these are the first results for the problem
of sparsely and exactly representing matrices by weighted rectangles.
\end{abstract}

\section{Introduction}
This paper studies two related problems of ``explaining" data parsimoniously.
In the first part of this paper,
we focus on providing a top-down ``hierarchical explanation''
of ``tree-ordered'' matrix data.
We motivate the problem as follows.
Suppose that one is given a matrix $A=(a_{ij})$ of data,
and that the rows naturally correspond to the leaves of a rooted tree $T_1$,
and the columns, to the leaves of a rooted tree $T_2$.  
For example, $T_1$ and $T_2$ could represent hierarchical IP addresses
spaces with nodes corresponding to IP prefixes.
Each node of either $T_1$ or $T_2$ is then said to correspond
to the set of rows (or columns, respectively)
corresponding to its leaf descendants.
Say {\tt 128.*} (i.e., the set of all $2^{24}$ 
IP addresses beginning with ``128'', which happens to correspond to
the {\tt .edu} domain) is a node in $T_1$ and {\tt 209.85.225.*} 
(i.e., the set of all $2^8$ IP addresses beginning with {\tt
209.85.225},
which is {\tt www.google.com}'s domain) is a node in $T_2$.
Then {\tt (128.*, 209.85.225.*)} could, say, represent the amount of traffic
%(assuming the matrix entries $a_{ij}$ summarize, say, the
%number of packets between entities $i$ and $j$) 
flowing from
all hosts in the {\tt .edu} domain (e.g., {\tt 128.8.127.3}) to all hosts in the {\tt www.google.com}
domain (e.g., {\tt 209.85.225.99}).
It is easy to relabel the rows or columns so that each internal node 
of $T_1$ or $T_2$ corresponds to a contiguous block of rows or columns.
%(by numbering the edges from a node to its children
%$1,2,\ldots$, and then sorting the leaves in
%lexicographic order by the sequence of labels on the path from
%the root to the leaf).

We need a few definitions.
Let us say a {\em rectangle} in an $m\times n$ matrix $A$ is a set 
$Rect(i_1,i_2,j_1,j_2)=\{i:i_1\le i \le  i_2\}\times \{j:j_1\le j\le j_2\}$,
for some $1\le i_1\le i_2\le m$, $1\le j_1\le j_2\le n$.  Certain rectangles are {\em
allowed};  others are not.  Let $\cal R$ denote the set of allowed rectangles.
Say a set of $w(R)$-weighted rectangles $R$ {\em represents $A=(a_{ij})$} if 
for any cell $(i,j)$, the sum of $w(R)$ over cells that contain $(i,j)$ is 
$a_{ij}$.

Returning to the Internet example, 
a pair $(u,v)$, $u$ a node of $T_1$, $v$ a node of $T_2$, corresponds to a 
rectangle.   Say that a rectangle is allowed, relative to $T_1$ and $T_2$, if 
it is the cross product of the set of rows corresponding to some node $u$ in $T_1$
and the set of columns corresponding to some node $v$ in $T_2$.  In this scenario, we attempt to ``explain'' or ``describe'' the matrix 
by writing it as a sum of weighted allowed rectangles. 
Formally, we wish to assign a weight $w_R$ to each allowed rectangle $R$ such that the set of weighted rectangles represents $A$.

Of course there is always a solution: one can just assign weights to the $1 \times 1$ rectangles.
But this is a trivial description of the matrix.  Usually more concise explanations are preferable.
For this reason we seek an ``explanation'' with as few nonzero terms as possible.  More precisely, 
we seek to assign a weight $w_R$ to each allowed rectangle $R$ such that the set of weighted rectangles represents $A$, and
such that the number of nonzero weights $w_R$ assigned is minimized.
(We define problems formally in Section \ref{dande}.)

%of matrix data as a linear combination of rectangles in the
%cross-product of attribute tree hierarchies, analogous to that
%proposed in~\cite{ABGYKS} for single-attribute hierarchies.
Here is a 1-dimensional example.
Suppose that a media retailer sells items in exactly four categories:  action-movie DVD's, 
comedy DVD's, books, and CD's.  The retailer builds a
hierarchy with four leaves, one for each of the categories of items.  A node ``DVD's" is the parent of leaves ``action-movie DVD's" and ``comedy
DVD's''.  There is one more node, a root labeled ``all'', with children ``DVD's", ``books'', and ``CD's''.  

Suppose that one year, sales of action-movie DVD's increased by \$6000 and sales of the other three categories increased by \$8000 each.
One could represent the sales data by giving those four numbers, one for each leaf of the hierarchy,  
yet one could more parsimoniously say that there was a general increase
of \$8000 for all (leaf) categories, in addition to which there was a decrease of \$2000 for action-movie DVD's.  This is represented by assigning
\$8000 to node ``all'' and \$-2000 to ``action-movie DVD's''.
While many different linear combinations may be possible,
simple explanations tend to be most informative.
Therefore, we seek an answer minimizing the explanation size
(the number of nonzero terms required in the explanation).

Here is a definition of {\sc Tree$\times$Tree}.
An instance 
consists of an $m\times n$ matrix $A=(a_{ij})$, along
with two rooted trees, a tree $T_1$  whose leaf set is the set of rows of
the matrix, and a tree $T_2$ whose leaf set is the set of
columns.  
Let $L_i(v)$ be the leaf descendants of node $v$ in tree $T_i$,
$i \in \{1,2\}$.  
Now $\cal R$ is just the set $\{L_1(u) \times L_2(v) \ : \ \mbox{$u$
is a node in $T_1$ and $v$ is a node in $T_2$}\}$.
The goal is to find the smallest set of weighted rectangles which represents $A$.
We prove this problem NP-hard and give a randomized 2-approximation
algorithm for it.  APX-hardness is not known.

The second problem, {\sc AllRects}, is motivated by the need to concisely describe or explain linearly-ordered data.  
Imagine that one has two ordered parameters,
such as horizontal and vertical location, or age and salary.
No trees are involved now.   Instead we allow any interval of rows
(i.e., $\{i: i_1\le i\le i_2\}$ for any $1\le i_1\le i_2\le m$) 
and any interval of columns
(i.e., $\{j: j_1\le j\le j_2\}$ for any $1\le j_1\le j_2\le n$).
For example, $[800,1000]\times [500,1500]$ could be used to
represent a geographical region extending eastward from 800 to 1000 miles
and northward from 500 to 1500 miles, and
$[35.0,45.0]\times [80000,95000]$ could be used to represent the
subset of people 35-44 years old and earning a salary of \$80000-\$95000.
Then we can use the former ``rectangles'' to summarize the change
(say, in population counts) with respect to location,
or use the latter with respect to demographic attributes age and salary. 
%Using the latter for an example, suppose the population
%increased by 25,000 in $[35,45]\times [80000,95000]$ but
%decreased by 100,000 in $[25,40]\times [70000,90000]$.
%Therefore, the net change in $[35,40]\times [80000,90000]$
%would be -75,000. 

%into a sum of rectilinear rectangles.
%For example, these images could have been generated by a
%computer drawing program such as {\tt Xfig}, or by a
%CAD-CAM system.
%Here we allow all rectangles, i.e., the allowed set 
%${\cal R}=\{Rect(i_1,i_2,j_1,j_2)|
%1\le i_1\le i_2\le m, 1\le j_1\le j_2\le n\}$, presuming $A$
%is $m\times n$.  

Hence in {\sc AllRects}
the set $\cal R$ of allowed rectangles is the cross product
between the set of row intervals and the set of column intervals.
As a linear combination
of how few arbitrary rectangles can we write the given
matrix?   We prove this problem NP-hard and give a
2.56-approximation algorithm for it.  Again, APX-hardness is
unknown.

\section{Related Work}
To our knowledge, while numerous papers have studied similar
problems, none proposes any algorithm for either of the
two problems we study.  
One very relevant prior piece of work  is a polynomial-time
exact algorithm solving 
the 1-dimensional version of {\sc Tree$\times$Tree} (more
properly called the ``tree'' case in 1-d, since only one tree is
involved) \cite{ABGYKS}.  Here, as in the media-retailer example above,
we have a sequence of integers and a tree whose leaves are the elements of the sequence.
Indeed, we use this algorithm heavily in
constructing our randomized constant-factor approximation algorithm for the 
tree$\times$tree case.  
%(In fact, \cite{ABGYKS} exactly solves also an inexact version of the one-dimensional problem in which
%the weighted sum of rectangles containing cell $(i,j)$ is required to differ from $a_{i}$ by at most $\Delta$
%in absolute value.)

Relevant to our work is \cite{BCS} by Bansal, Coppersmith, and
Schieber,
which (in our language) studies
the 1-d (exact) problem in which all intervals are allowed and all must have {\em nonnegative} weights,
proves the problem NP-hard, and
gives a constant-factor approximation algorithm.

Also very relevant is a paper by
Natarajan~\cite{Natarajan}, which studies an ``inexact'' version
of the problem: instead of finding weighted rectangles whose
sum of weights is $a_{ij}$ exactly, for each matrix cell $(i,j)$,
these sums approximate the $a_{ij}$'s.
(Natarajan's algorithm is more general and can
handle any arbitrary set $\cal R$ of allowed rectangles;
however, the algorithm is very slow.)
More precisely, in the output set of
rectangles, define $a'_{ij}$ to be the sum of the weights of the
rectangles containing cell $(i,j)$.  Natarajan's algorithm 
ensures, given a tolerance $\Delta>0$, that the $L_2$ error 
$\sqrt {\sum_{i=1}^m \sum_{j=1}^n (a'_{ij}-a_{ij})^2}$ 
is at most $\Delta$.
(Natarajan's algorithm cannot be used for $\Delta=0$.)
The upper bound on the number of
rectangles produced by Natarajan's algorithm is a factor of
approximately $18 \ln (||A||_2/\Delta)$ (where $||A||_2$
is the square root of the sum of squares of the entries of $A$) 
larger than the optimal number used by an adversary who is
allowed, instead, only $L_2$-error $\Delta/2$.  
% This is because the $L_2$ norm of the quasi-inverse is 1, since the matrix $A$
% contains an identity matrix.
Furthermore, Natarajan's algorithm is very slow, much slower than our algorithms. 
%%See the full version of our paper for details.

Frieze and Kannan in \cite{FK} show how to inexactly represent a matrix as a sum of a
small number of rank-1 matrices, but their method is unsuitable to solve our
problem, as not only is there no way to restrict the
rank-1 matrices to be rectangles, the error is of $L_1$ type rather than
$L_\infty$.  In other words, the {\em sum} of the $mn$ errors is bounded by $\Delta
mn$, rather than individual errors' being bounded by $\Delta$.

Our problem may remind readers of compressed sensing, the decoding
aspect of which requires one to seek a solution $x$ with fewest
nonzeroes to a linear system $Hx=b$.
The key insight of compressed sensing is that when $H$ satisfies the
``restricted isometry property" \cite{W,CT,Donoho}, as do almost all random
matrices, the solution $x$ of minimum $L_1$ norm is also the sparsest
solution.  The problem with applying compressed sensing to the problems mentioned
herein, when the matrix $A$ is $m\times n$, is that the associated matrix $H$, which has $mn$ rows and a number of
columns equal to the number of allowed rectangles, is anything but random.  
On a small set of test instances, the
authors found the solutions of minimum $L_1$ norm
(using linear programming) and discovered that 
they were far from sparsest.

Other authors have studied other ways of representing matrices.
%Bu et al.~\cite{Bu05} studied the problem of representing a
%{\em binary} matrix as a union of rectangles (from the tree
%cross-product) whose entries are all 1, followed by the set
%difference of a union of rectangles whose entries are all 0.
Applegate et al.~\cite {ACJKLW}
studied the problem of representing a {\em binary} matrix,
starting from an all-zero matrix, by an {\em ordered} sequence of rectangles,
each of whose entries is all 0 or all 1, in which $a_{ij}$ should equal
the entry of the {\em last} rectangle which contains cell $(i,j)$.
Anil Kumar and Ramesh \cite{AKR} study the same model in which only
all-1 rectangles are allowed (in which case the order clearly
doesn't matter).
Two papers \cite{Roth08,Karpinski09} study the Gale-Berlekamp
switching game and can be thought of as a variant of our problem
over $\mathbb{Z}_2$.

\section{A Few Words About Practicality}
Admittedly, for noisy data in the real world,
probably more practical problems than our ``exact'' 
problems are these two bounded-error  
(i.e., $L_\infty$) ``inexact'' problems:
Given an input of either {\sc Tree$\times$Tree} or {\sc AllRects}
and a number $\Delta\ge 0$, find a smallest subset of allowed rectangles, and weights for each,
such that for any cell $(i,j)$, $a_{ij}$ differs from the sum of the weights 
of the rectangles containing $(i,j)$ by at most $\Delta$ in absolute value.  
%Unfortunately the present authors have no constant-factor approximation algorithms for these inexact 
problems and so we leave them for future work.  
Nonetheless,
we find the exact problems interesting and the solutions nontrivial,
and hope that studying them may yield insight for solving the $\Delta>0$ case.

\section{Formal Definitions and Examples}\label{dande}
Given an $m\times n$ matrix $A=(a_{ij})$ and $1\le i_1\le i_2\le m$,
$1\le j_1\le j_2\le n$, recall that $Rect(i_1,i_2,j_1,j_2)=
\{(i,j)|i_1\le i\le i_2,j_1\le j\le j_2\}$.
Define $Rects=\{Rect(i_1,i_2,j_1,j_2)|1\le i_1\le i_2\le m,
1\le j_1\le j_2\le n\}$.
For each of the two problems, we are given a subset ${\cal R}\subseteq
Rects$;
the only difference between the two problems we discuss is the definition of $\cal R$.  
The goal is to find a smallest subset $OPT_2(A)$ of
${\cal R}$, and an associated
weight $w(R)$ (positive or negative) for each rectangle $R$, 
such that every cell
$(i,j)$ is covered by rectangles whose weights sum to $a_{ij}$,
that is, 
\begin{equation} \label {eqn1}
a_{ij}=\sum_{R: 
R \in OPT_2(A)
\mbox{ and }
R\ni (i,j) }
w(R),
\end{equation}
the ``2'' in ``$OPT_2(A)$'' referring to the fact that $A$ is
2-dimensional.

While the algorithm for the 
tree$\times$tree case appears (in Section \ref{treextree})
before that for the arbitrary-rectangles case (in Section
\ref{arbitrary}), here we define {\sc AllRects}, the latter, first, since it's
easier to define.
As mentioned above, we call the case of ${\cal R}=Rects$ {\sc AllRects}.

\noindent {\bf Example.}
Since the matrix
\begin{eqnarray*}
A = \left[ \begin{smallmatrix}
2 & 2 & 2 & 2 \\
5 & 3 & 1 & 2 \\
6 & 4 & 1 & 3 \\
5 & 5 & 2 & 2
\end{smallmatrix} \right]
&=&
2\left[ \begin{smallmatrix}
1 & 1 & 1 & 1 \\
1 & 1 & 1 & 1 \\
1 & 1 & 1 & 1 \\
1 & 1 & 1 & 1
\end{smallmatrix} \right]
+
3\left[ \begin{smallmatrix}
0 & 0 & 0 & 0 \\
1 & 1 & 0 & 0 \\
1 & 1 & 0 & 0 \\
1 & 1 & 0 & 0
\end{smallmatrix} \right]
+
1\left[ \begin{smallmatrix}
0 & 0 & 0 & 0 \\
0 & 0 & 0 & 0 \\
1 & 1 & 1 & 1 \\
0 & 0 & 0 & 0
\end{smallmatrix} \right]
-
2\left[ \begin{smallmatrix}
0 & 0 & 0 & 0 \\
0 & 1 & 1 & 0 \\
0 & 1 & 1 & 0 \\
0 & 0 & 0 & 0
\end{smallmatrix} \right]
+
1\left[ \begin{smallmatrix}
0 & 0 & 0 & 0 \\
0 & 0 & 1 & 0 \\
0 & 0 & 0 & 0 \\
0 & 0 & 0 & 0
\end{smallmatrix} \right],
\end{eqnarray*}
$A$ can be written as a linear combination with
$w(\{1,2,3,4\}\times \{1,2,3,4\})=2$, 
$w(\{2,3,4\}\times \{1,2\})=3$,
$w(\{3\}\times \{1,2,3,4\})=1$, 
$w(\{2,3\}\times \{2,3\})=-2$, and
$w(\{2\}\times \{3\})=1$. 
Hence $|OPT_2(A)|\le 5$. 

%The problem {\sc Tree$\times$Tree} is slightly trickier to define,
We need some notation in order to define
{\sc Tree$\times$Tree}, in which we are also given trees $T_1$
and $T_2$.
%Given an $m\times n$ matrix $A$, 
%we are given two trees, $T_1$, which has one
%leaf for every row in $A$, and $T_2$,
%with one leaf for every column.  
We use $R_i$ to denote the row vector in the $i$th row of
the input matrix, $1\le i\le m$.
%We say that a leaf node $l$ of $T_1$ {\em corresponds} to row $R_l$
%and say that a nonleaf node $u$ of $T_1$ {\em corresponds} to
%the set of rows
For a node $u\in T_1$, let $S^1_u = \{R_l \ : \ l\mbox { is a
leaf descendant in $T_1$ of }u\}$.
Similarly, we
use $C_j$ to denote the column vector in the $j$th column of
the input matrix, $1\le j\le n$.
%We say that a leaf node $l$ of $T_2$ {\em corresponds} to column $C_l$
%and say that a nonleaf node $v$ of $T_2$ {\em corresponds} to
%the set of columns
For a node $v\in T_2$, let $S^2_v = \{C_l \ : \ l\mbox { is a
leaf descendant in $T_2$ of }v\}$.
Note that, since $T_1$ and $T_2$ are trees,
$\{S^1_u| u\in T_1\}$ and $\{S^2_v|v\in T_2\}$ are laminar.

In this notation,
in {\sc Tree$\times$Tree}, 
${\cal R}=\{S^1_u|u\in T_1\}\times \{S^2_v|v \in T_2\}$.

%we are given a real-valued
%$m\times n$ matrix $A=(a_{ij})$, a tree $T_1$, which has one leaf
%for every row in the matrix, and a tree $T_2$, with one leaf for
%every column.
%The set $\cal R$ is the set of all rectangles $R$ such that
%$R=S^1_u\times S^2_v$ for some $u\in T_1$ and $v\in T_2$.  
%In other words,
%our goal is to find the smallest set $OPT_2(A)$ of rectangles $R$,
%and associated (positive or negative) weights $w(R)$, 
%such that $R=S^1_u\times S^2_v$ for some $u\in T_1$ and $v\in T_2$.  
%$Rect(i_1,i_2,j_1,j_2) = \{i,j: i_1\leq i \leq i_2, j_1\leq j \leq j_2\}$ of the form 
%$S^1_u\times S^2_v$ (where $u$ is a node of $T_1$ and $v$ is a node of $T_2$)
%with associated (non-zero) real weights $w(i_1,i_2,j_1,j_2)$,
%such that every cell $(i,j)$ is covered by rectangles whose
%weights sum to total weight  $a_{ij}$. That is,
%\begin{equation}
%\sum_{i_1,i_2,j_1,j_2: (i,j) \in Rect(i_1,i_2,j_1,j_2)} w(i_1,i_2,j_1,j_2) = a_{ij}.
%\end{equation}
%\end{definition}

\noindent {\bf Example.}
Using trees $T_1$, $T_2$ having a root with four children (and
no other nodes) apiece, 
we may use any single row or all rows, and any single column or
all columns.
%shown in Figure~\ref{fig:trees},
For example, since the matrix
\begin{eqnarray*}
A = \left[ \begin{smallmatrix}
5 & 3 & 4 & 5 \\
3 & 0 & 2 & 4 \\
2 & 2 & 1 & 3 \\
3 & 3 & 2 & 3
\end{smallmatrix} \right]
&=&
3\left[ \begin{smallmatrix}
1 & 1 & 1 & 1 \\
1 & 1 & 1 & 1 \\
1 & 1 & 1 & 1 \\
1 & 1 & 1 & 1
\end{smallmatrix} \right]
+
2\left[ \begin{smallmatrix}
1 & 1 & 1 & 1 \\
0 & 0 & 0 & 0 \\
0 & 0 & 0 & 0 \\
0 & 0 & 0 & 0
\end{smallmatrix} \right]
-
1\left[ \begin{smallmatrix}
0 & 0 & 0 & 0 \\
0 & 0 & 0 & 0 \\
1 & 1 & 1 & 1 \\
0 & 0 & 0 & 0
\end{smallmatrix} \right]
-
1\left[ \begin{smallmatrix}
0 & 0 & 1 & 0 \\
0 & 0 & 1 & 0 \\
0 & 0 & 1 & 0 \\
0 & 0 & 1 & 0
\end{smallmatrix} \right]\\
&-&
2\left[ \begin{smallmatrix}
0 & 1 & 0 & 0 \\
0 & 0 & 0 & 0 \\
0 & 0 & 0 & 0 \\
0 & 0 & 0 & 0
\end{smallmatrix} \right]
-
3\left[ \begin{smallmatrix}
0 & 0 & 0 & 0 \\
0 & 1 & 0 & 0 \\
0 & 0 & 0 & 0 \\
0 & 0 & 0 & 0
\end{smallmatrix} \right]
+
1\left[ \begin{smallmatrix}
0 & 0 & 0 & 0 \\
0 & 0 & 0 & 1 \\
0 & 0 & 0 & 0 \\
0 & 0 & 0 & 0
\end{smallmatrix} \right]
+
1\left[ \begin{smallmatrix}
0 & 0 & 0 & 0 \\
0 & 0 & 0 & 0 \\
0 & 0 & 0 & 1 \\
0 & 0 & 0 & 0
\end{smallmatrix} \right],
\end{eqnarray*}
we can write $A$ as a sum with
$w(\{1,2,3,4\}\times \{1,2,3,4\})=3$, 
$w(\{1\}\times \{1,2,3,4\})=2$, 
$w(\{3\}\times \{1,2,3,4\})=-1$, 
$w(\{1,2,3,4\}\times \{3\})=-1$, 
$w(\{1\}\times \{2\})=-2$, 
$w(\{2\}\times \{2\})=-3$, 
$w(\{2\}\times \{4\})=1$, and
$w(\{3\}\times \{4\})=1$. 
Since there are eight matrices, $|OPT_2(A)|\le 8$.

%$w(1,4,1,4)=3$, $w(1,1,1,4)=2$, $w(3,3,1,4)=-1$, $w(1,4,3,3)=-1$,
%$w(1,1,2,2)=-2$, $w(2,2,2,2)=-3$, $w(2,2,4,4)=1$ and $w(3,3,4,4)=1$
%corresponding to rectangles in the cross-product $T_1 \times T_2$.
%Grouping rectangles of the same width and length into a single matrix
%using these weights, $A$ can be written as

Note that we use the same notation, $OPT_2(A)$, for the optimal
solutions of both {\sc AllRects} 
and {\sc Tree$\times$Tree}.

\section{Approximation Algorithm for {\sc Tree\texorpdfstring{$\times$}{ x }Tree}}\label{treextree}

We defer the proof of NP-Hardness of {\sc Tree$\times$Tree} to the appendix.

Our algorithm
will rely upon the exact %[kostya] was: optimal 
algorithm, 
due to Agarwal et al.~\cite{ABGYKS}, for the case in which
the matrix has just one column (that is, the 1-dimensional
case).  
\begin{definition}
Given a fixed rooted tree $T_1$ with $m$ leaves,
and an $m$-vector $V=(v_i)$, let
$OPT_1(V)$ denote a smallest set of intervals
$I=\{i:i_1\le i\le i_2\}\subseteq [1,m]$ and associated weights $w(I)$, 
each $I$ corresponding to a node of
$T_1$, such that for all $i$, $v_i=
\sum_{I: I \in OPT_1(V) \mbox{ and }I\ni i} w(I)$.
\end{definition}

Clearly $|OPT_1(V)|$ equals $|OPT_2(V')|$, where $V'$ is the
$m\times 1$ matrix containing $V$ as a column.
The difference is that $OPT_1(V)$ is a set of vectors while 
$OPT_2(V')$ is a set of rectangles.
We emphasize that $V$ is a vector and that the definition
depends on $T_1$ and not $T_2$ by putting the ``1'' in
``$OPT_1(V)$''.  The key point is that 
\cite{ABGYKS} showed how to compute $OPT_1(V)$ exactly.  
%rectangles and weights, by $Solution_1(V)$.

In order to charge the algorithm's cost against $OPT_2(A)$, we need
to know some facts about $OPT_2(A)$.  
Recall that $OPT_2(A)$ is a smallest subset of $\cal R$ such that there are 
weights $w(R)$ such that equation (\ref{eqn1}) holds.  

\begin{definition}

\noindent
\begin{enumerate}
\item
For each rectangle $R$ and associated weight $w_R$, 
let $R'_{w_R}$ denote the $m \times n$ matrix which is 0 for every cell $(i,j)$,
except that ${R'_{w_R}}_{ij}:=w_R$ if $(i,j)\in R$.
\item
Given a vertex $v$ of $T_2$, let $D_v$ be the set of all $R\in OPT_2(A)$ 
such that $R$ has column set exactly equal to $S_v^2$.   
\item
Now let $K_v=\sum_{R \in D_v} R'_{w_R}$.  By definition of $D_v$, all columns $j$
of $K_v$
for $j\in D_v$ are the same.  Let $V_v$ be column $j$ of $K_v$ for any $j\in D_v$.
\end{enumerate}
\end{definition}

\begin{lemma}\label{mainforalgo}
The column vectors $(V_v)$ satisfy the following:
\begin{enumerate}
\item \label{part1}
For all leaves $l$ in $T_2$, the vector $C_l$ 
equals the sum of $V_v$ over all ancestors $v$ of $l$ in $T_2$.

\item \label{part1prime}
For all leaves $l'$ and $l''$ in $T_2$ with a common ancestor $u$,
the vector $C_{l'} - C_{l''}$ equals the sum of $V_v$ over all vertices $v$
on the path from $u$ down to $l'$ (not including $v=u$) minus the sum of $V_v$
over all vertices $v$ on the path from $u$ down to $l''$ (not including $v=u$).

\item \label{part2new} 
The union, over all vertices $v\in T_2$, of 
$OPT_1(V_v)\times \{S^2_v\}$ (which 
obviously has size $|OPT_1(V_v)|$), with the corresponding
weights,
is an optimal solution for {\sc Tree$\times$Tree} on $A$.

\item \label{part2}
$|OPT_2(A)|=\sum_{v \in T_2} |\o1(V_v)|$.
\end{enumerate}
\end{lemma}
\begin{proof}
The nodes $v$ which correspond to sets of columns containing column $C_l$
are exactly the ancestors in $T_2$ of $l$.
Hence, Part \ref{part1} follows. 

Part~\ref{part1prime} is an immediate corollary of Part \ref{part1}.

Clearly, by Part \ref{part1}, the union over all vertices $v\in T_2$ of 
$OPT_1(V_v)\times \{S^2_v\}$
is a feasible solution for {\sc Tree$\times$Tree} on $A$.
It is also optimal, and here is a proof.  
The size of the optimal solution $OPT_2(A)$
equals the sum, over vertices $v\in T_2$, of the number of
rectangles in $OPT_2(A)$ having column set $S^2_v$.  Fix a
vertex $v\in T_2$.  Since the weighted sum of the rectangles in
$OPT_2(A)$ with column set $S^2_v$ is $V_v$, and each has a row
set $S^1_u$ for some $u\in T_1$, the number of such
rectangles must be at least $OPT_1(V_v)$.  If the number of rectangles
with column set $S^2_v$ strictly exceeded $OPT_1(V_v)$, we could
replace all rectangles in $OPT_2(A)$ having column set $S^2_v$ by
a smaller set of weighted rectangles having column set $S^2_v$, each of whose 
columns is the same, and summing to $V_v$
in each column;
since the new set and the old set have the same
weighted sum, the new solution would still sum to $A$, and have
better-than-optimal size, thereby contradicting optimality of
$OPT_2(A)$. Part \ref{part2new} follows.

Part~\ref{part2} follows from Part \ref{part2new}.
\end{proof}

Lemma \ref{mainforalgo} will be instrumental in analyzing the
algorithm. 

While the algorithm is very simple to state, it was nontrivial
to develop and analyze.
In the algorithm, we use the algorithm by Agarwal et al.~\cite{ABGYKS}
to obtain $OPT_1(V)$ given a vector $V$.
%the exact one dimensional solution.

\rule{0pt}{12pt}
\hrule height 0.8pt
\rule{0pt}{1pt}
\hrule height 0.4pt
\rule{0pt}{6pt}

\textbf{Algorithm for {\sc Tree\texorpdfstring{$\times$}{ x }Tree}}
\begin{enumerate}
\item For every internal node $u$ in the tree $T_2$, pick a random child $u^*$ of $u$ and let  $c(u) = u^*$. Let $path(u)$ be the random path going from $u$ to a leaf:
$$u \mapsto c(u) \mapsto c(c(u)) \mapsto \cdots \mapsto l(u),$$
where we denote the last node on the path, the leaf, by $l(u)$.

\item Where $root$ denotes the root of $T_2$, 
for every node $u$ in $T_2$, in 
increasing order by depth, do:
\begin{itemize}
\item If $u$ is the root of $T_2$, then
\begin{itemize}
%\item  Output $OPT_1(C_{l(root)})$ rectangles from $Solution_1(C_{l(root)})\times S^2_{root}$.
\item  Output $OPT_1(C_{l(root)})\times \{S^2_{root}\}$ with the
corresponding weights (those of the optimal solution for
$C_{l(root)}$).
%Cover the whole matrix with $OPT_1(C_{l(u)})$  rectangles with the sum of each column equal to $C_{l(u)}$.
\end{itemize}
\item Else
\begin{itemize}
\item Let $p(u)$ be the parent of $u$.
\item  Output $OPT_1(C_{l(u)} - C_{l(p(u))})\times \{S^2_u\}$
with the corresponding weights.
\end{itemize}
\end{itemize}
\end{enumerate}
\rule{0pt}{1pt}
\hrule height 0.4pt
\rule{0pt}{1pt}
\hrule height 0.8pt
\rule{0pt}{12pt}

%\pagebreak

\begin{theorem}\label{thm:algmain}
The expected cost of the algorithm is at most $2|OPT_2(A)|$.
\end{theorem}

In the main part of the paper we prove a weaker guarantee
for exposition: the expected cost of the algorithm is at most $4|OPT_2(A)|$.
We defer the improvement to the appendix.

The algorithm can be easily derandomized using dynamic programming.

\begin{proof}
Every column $C_u$ is covered by rectangles with sum
$$(C_u - C_{l(p(u))}) + (C_{l(p(u))} - C_{l(p(p(u)))}) + \cdots + C_{l(root)} = C_u.$$
Thus the algorithm produces a valid solution. 
We now must estimate the expected cost of the solution. 
The total cost incurred by the algorithm is
$$|OPT_1(C_{l(root)})|+\sum_{u\ne root}
|OPT_1(C_{l(u)}-C_{l(p(u))})|.$$
Assume, without loss of generality, that all nodes in the tree  either have two or more children or are leaves. Denote the number of children of a node $v$, the degree of $v$,
by $d(v)$. Denote by $\ONE$ the indicator function. Observe that for the root node we have
$$|OPT_1(C_{l(root)})| = \left |OPT_1 \left(\sum_{v \in path(root)} V_v\right)\right | \leq \sum_{v \in path(root)} |OPT_1(V_v)|;$$
for a nonroot vertex $u$, we have by Lemma~\ref{mainforalgo} \eqref{part1prime},
keeping in mind that $l(\cdot)$, $c(\cdot)$, and $path(\cdot)$ are random,
\begin{eqnarray*}
|OPT_1(C_{l(u)} - C_{l(p(u))})| &=& \left |OPT_1\left( \sum_{v\in path(u)} V_v - \sum_{v\in path(c(p(u)))} V_v\right)\right |\\
&\leq& \left(\sum_{v\in path(u)} |OPT_1 (V_v)| + \sum_{v\in path(c(p(u)))} |OPT_1 (V_v)| \right)\cdot\ONE (u \neq c(p(u))).
\end{eqnarray*}
Here we used the triangle inequality for the function $|OPT_1(\cdot)|$.

Consider the second sum in the right-hand side. For every child $u'$ of $p(u)$, 
the random node $c(p(u))$ takes value $u'$ with probability
$1/d(p(u))$. Thus 
\begin{align*}
\Exp\Big[\sum_{v\in path(c(p(u)))}& |OPT_1 (V_v)| \cdot\ONE (u \neq c(p(u)))\Big] \\
=
\frac{1}{d(p(u))}&\sum_{u': u' \text{ is a sibling of } u}\E{
\Big(\sum_{v\in path(c(p(u)))} |OPT_1 (V_v)|\Big)\ \mid c(p(u))=u'}\\
=
\frac{1}{d(p(u))}&\sum_{u': u' \text{ is a sibling of } u}\E{\sum_{v\in path(u')} |OPT_1 (V_v)|}.
\end{align*}
%We implicitly used that the random events $\{v\in path(u')\}$ and $\{u' = c(p(u)))\}$ are independent.
%Similarly, the events $\{v\in path(u)\}$ and $\{u \neq c(p(u)))\}$ are independent. 
$\Prob{u \neq c(p(u))}$ equals $(d(p(u)) - 1)/d(p(u))$. Denote this expression by $\alpha_u$. 
The total expected size of the solution returned by the algorithm is bounded by
\begin{align} \textstyle
\Exp&\Big[\sum_{v\in path(root)} |OPT_1 (V_v)| \Big] +  \sum_{u\ne root} \alpha_u \E{\sum_{v\in path(u)} |OPT_1 (V_v)| } \\
&\qquad\qquad{\textstyle\phantom{=}
+ \sum_{u\ne root}\frac{1}{d(p(u))}\sum_{u': u' \text{ is a sibling of } u}\E{\sum_{v\in path(u')} |OPT_1 (V_v)|}}\nonumber\\
&\qquad{\textstyle= \E{\sum_{v\in path(root)} 
|OPT_1 (V_v)| } + \sum_{u\ne root}\alpha_u \E{\sum_{v\in path(u)} |OPT_1 (V_v)| }}
\nonumber\\
&\qquad\qquad{\textstyle \phantom{=} + \sum_{u'\ne
root}
\Big(\sum_{u\ne root}\frac{\ONE (u' \text{ is a sibling of } u)}{d(p(u'))}\Big)\E{\sum_{v\in path(u')} |OPT_1 (V_v)|}.}\label{eq:TotalBound}
\end{align}
%For the root vertex $u=root$ we assume that $c(p(u))$ is not defined and thus $\ONE(c(p(u))\neq u) = 1$. 
Notice that, for a fixed $u'\ne root$,
\begin{equation}
\sum_{u\ne root}\frac{\ONE (u' \text{ is a sibling of } u)}{d(p(u'))} = \frac{d(p(u')) - 1}{d(p(u'))} = \alpha_{u'}< 1. \label{eq:SumU}
\end{equation}
Hence, the total cost of the solution is bounded by
$$\!\sum_{u}\E{\!\sum_{v\in path(u)}\!\! |OPT_1 (V_v)| } + \sum_{u'\ne
root}\E{\!\sum_{v\in path(u')}\!\! |OPT_1 (V_v)|} \le 
2\sum_{u}\E{\!\sum_{v\in path(u)}\!\!|OPT_1 (V_v)| }.
$$

Finally, observe that node $v$ belongs to $path(v)$ with probability $1$; it belongs to the $path (p(v))$ with probability
at most $1/2$; it belongs to the path $path (p(p(v)))$ with probability at most $1/4$, etc. It belongs to $path(u)$ with probability 0 if $u$ is
not an ancestor of $v$.  Thus
\begin{eqnarray*}
2\sum_{u}\E{\sum_{v\in path(u)} |OPT_1 (V_v)|} &=& 2\sum_{v} |OPT_1 (V_v)| \cdot \Big(\sum_u \Prob{v \in path(u)}\Big) \\&\leq& 
2 \sum_{v} |OPT_1 (V_v)|  \cdot \Big(1 + 1/2 + 1/4 + \cdots \Big) \\&<& 4\sum_{v} |OPT_1 (V_v)| \leq 4|OPT_2(A)|.
\end{eqnarray*}
We have proven that the algorithm finds a 4-approximation. 
A slightly more careful analysis, in the appendix, shows that the approximation ratio of the algorithm is at most 2.
\end{proof}

What is the running time of the 2-approximation algorithm?  The time needed to run the 1-dimensional algorithm of 
\cite{ABGYKS} is $O(dn)$ where there are $n$ leaves in each tree
and the {\em smaller} of the two depths is $d$.  
One can verify that the running time of our 2-approximation algorithm is a factor $O(n)$ larger, or 
$O(dn^2)$.  In most applications at least one of the trees would have depth $O(\log n)$, giving $O(n^2 \log n)$ in total.

\section{Approximation Algorithm For {\sc AllRects}}\label{arbitrary}
%\begin{definition}
%In the {\sc AllRects} problem, we are given a
%real-valued matrix $A=(a_{ij})$, and our goal is to 
%find the smallest set of rectangles 
%$Rect(i_1,i_2,j_1,j_2) = \{i,j: i_1\leq i \leq i_2, j_1\leq j \leq j_2\}$ and real weights
%$w(i_1,i_2,j_1,j_2)$ such that every cell $(i,j)$ is covered by rectangles
%whose weights sum to total weight  $a_{ij}$. That is,
%\begin{equation}\label{eq:problemdef}
%\sum_{i_1,i_2,j_1,j_2: (i,j) \in Rect(i_1,i_2,j_1,j_2)} w(i_1,i_2,j_1,j_2) = a_{ij}.
%\end{equation}
%\end{definition}

\subsection{The 1-Dimensional Problem}
First we consider the one-dimensional case, for which we will 
give a $(23/18+\varepsilon)$-approximation algorithm;
$23/18<1.278$. We are given a sequence 
$a_1,a_2, \dots, a_{n}$ of numbers 
and we need to find a collection of closed  intervals $[i,j]$
with arbitrary real weights $w_{ij}$ so that every integral
point $k\in\{1,\ldots,n\}$ is covered by a set of intervals with
total weight $a_k$. That is, for all $k$,
\begin{equation}\label{eq:DefOneD}
\sum_{i,j: k\in[i,j]} w_{ij} = a_k.
\end{equation}
Our goal is to find the smallest possible collection. We shall use the approach of 
Bansal, Coppersmith, and Schieber~\cite {BCS} (in their problem all $a_i\geq 0$ and all $w_{ij}>0$). Set $a_0 = 0$ and $a_{n+1} = 0$. Observe that if $a_k=a_{k+1}$, then in the optimal solution every interval covering $k$ also covers $k+1$. On the other hand, since every rectangle covering both $k$ and $k+1$ contributes the same weight to $a_k$ and $a_{k+1}$, if $a_k \neq a_{k+1}$, then there should be at least one interval that either covers $k$ but not $k+1$, or covers $k+1$ but not $k$. By the same reason, the difference $a_{k+1}-a_k$, which we denote by $\Delta_k = a_{k+1}-a_k$, equals the difference between the weight of intervals with the left end-point at $k+1$ and the weight of rectangles with the right endpoint at $k$:
 \begin{equation}\label{eq:DeltaOneD}
\Delta_k = \sum_{j: k+1 \leq j} w_{k+1,j} - \sum_{i: i\leq k} w_{ik}.
\end{equation}
Note that if we find a collection of rectangles with weights
satisfying \eqref{eq:DeltaOneD}, then this collection of
intervals is a valid solution to our problem, i.e., then
equality \eqref{eq:DefOneD} holds. Define a directed graph on
vertices $\{0,\dots, n\}$. For every interval $[i,j]$, we add an
arc going from $i - 1$ to $j$. Then the condition
\eqref{eq:DeltaOneD} can be restated as follows: The sum of
weights of  arcs outgoing from $k$ minus the sum of weights of
arcs entering $k$ equals $\Delta_k$. Our goal is to find the
smallest set of arcs with non-zero weights satisfying this
property. Consider an arbitrary solution and one of the  weakly
connected components $S$. The sum  $\sum_{k\in S}\Delta_k = 0$,
since every arc is counted twice in the sum, once with the  plus sign and once with the minus sign. Since $S$ is a connected component the number of arcs connecting nodes in $S$ is at least $|S|-1$. Thus a lower bound on the number of arcs or intervals in the optimal solution is the minimum
of  
$$\sum_{t=1}^M (|S_t| - 1) = n+1 - M$$
among all partitions of the set of items $\{0,\dots, n\}$  into $M$ disjoint sets $S_1,\dots, S_M$ such that 
$\sum_{k\in S_t} \Delta_k = 0$
for all $t$.
On the other hand, given such a partition $(S_1,\dots, S_M)$, we can easily construct a set of intervals. Let $k_t$ be the minimal element in $S_t$. For every element $k$ in $S_t \setminus \{k_t\}$, we add an interval $[k_t + 1, k]$ with weight $-\Delta_{k}$. We now verify that these intervals satisfy~\eqref{eq:DeltaOneD}. If $k$ belongs to $S_t$ and $k\neq k_t$, then there is only one interval in the solution with right endpoint at $k$. This interval is $[k_t + 1,k]$ and its weight is $-\Delta_k$. The solution does not contain intervals with left endpoint at $k+1$ (since $k\neq k_t$). Thus ~\eqref{eq:DeltaOneD} holds as well. If $k$ belongs to $S_t$ and $k = k_t$, the solution does not contain intervals with the right endpoint at $k$, but for all
$k'\in S_t$ there is an interval $[k+1,k']$ with weight $-\Delta_{k'}$. The total weight of these intervals equals 
$$\sum_{k' \in S_t; k'\neq k} - \Delta_{k'} = - \sum_{k' \in S_t}\Delta_{k'} + \Delta_{k} = \Delta_k.$$
Condition ~\eqref{eq:DeltaOneD} again holds.

Thus the problem is equivalent to the problem of partitioning the set of items  $\{0,\dots, n\}$ into a family of $M$ sets 
$\{S_1, \dots, S_M\}$ satisfying the condition $\sum_{k\in S_t} \Delta_k = 0$ for all  $t$,
so as to minimize $\sum_{t} (|S_t| - 1) = (n+1) - M$.
Notice that the sum of all $\Delta_k$ equals 0. Moreover, every set with the sum of $\Delta_k$ 
equal to 0 corresponds to an instance of the 1-dimensional rectangle covering problem. We shall
refer to the problem as {\sc Zero-Weight Partition}.

We now describe the approximation algorithm for {\sc Zero-Weight
Partition} which is a modification of the algorithm of Bansal, Coppersmith, and Schieber~\cite {BCS} designed for a 
slightly different
problem  (that of minimizing setup times in radiation therapy). 
\begin{remark}
For {\sc Zero-Weight Partition}, our algorithm gives a slightly
better approximation guarantee than that of \cite {BCS}: $23/18\approx 1.278$ vs $9/7\approx 1.286$. The difference between algorithms is that the algorithm of Bansal, Coppersmith, and Schieber~\cite {BCS} performs either the first and third steps (in terms of our algorithm; see below), or the second and third steps; while our algorithm always performs all three steps.
\end{remark}

In the first step the algorithm picks all singleton sets $\{k\}$ with $\Delta_k = 0$ and pairs $\{i,j\}$ with $ \Delta_i= - \Delta_j$. It removes the items covered by any of the chosen sets. At the second step, with probability $2/3$ the algorithm enumerates all triples $\{i,j,k\}$ with $\Delta_i + \Delta_j+\Delta_k = 0$ and finds the largest 3-set packing among them using the $(3/2+\varepsilon)$-approximation algorithm due to Hurkens and Schrijver~\cite{HS}, i.e., it finds the largest (up to a factor of $(3/2+\varepsilon)$) disjoint family of triples $\{i,j,k\}$
with $\Delta_i + \Delta_j+\Delta_k = 0$.
Otherwise (with probability $1/3$), the algorithm enumerates all quadruples $\{i,j,k,l\}$ having $\Delta_i + \Delta_j+\Delta_k + \Delta_l = 0$ and finds the largest 4-set packing among them using the $(2+\varepsilon)$-approximation algorithm due to Hurkens and Schrijver~\cite{HS}. At the third, final, step the algorithm covers all remaining items, whose sum of $\Delta_k$'s is zero, with one set. 

Before we start analyzing the algorithm, let us consider a simple example. Suppose that $$(a_1, a_2,a_2,a_4,a_5,a_6) = (15,8,10,17,18,15).$$ First we surround the vector with two 0's:  $$(a_0, a_1, a_2,a_2,a_4,a_5,a_6, a_7)
 = (0,15,8,10,17,18,15,0).$$
Then compute the vector of $\Delta_k$'s: \begin{eqnarray*}(\Delta_0, \Delta_1, \Delta_2,\Delta_2,\Delta_4,\Delta_5,\Delta_6) &=& (15-0,8-15,10-8, 17 - 10 , 18-17, 15 - 18, 0- 15 )\\ &=& 
(15,-7,2,7 ,1,-3,-15).\end{eqnarray*}
Notice that $(-15) + 7 + (-2) + (-7) + (-1) + 3 + 15 = 0$. We partition the set into sets of weight~0: 
$$\{\Delta_0, \Delta_6\}, \{\Delta_1, \Delta_3\}, \{\Delta_2, \Delta_4, \Delta_5\}.$$
This partition corresponds to the following solution of the 1-dimensional problem:
interval $[1,6]$ with weight $15$, 
interval $[2,3]$ with weight $-7$, 
interval $[3,4]$ with weight $-1$, 
interval $[3,5]$ with weight $3$.

\begin{lemma}
For every positive $\varepsilon>0$, the approximation ratio of the algorithm when using $\varepsilon$ 
is at most $23/18 + O(\varepsilon)$, with $23/18 <1.278$.
\end{lemma}
\begin{proof}
First, observe that the partitioning returned by the algorithm is a valid partitioning, i.e., every item belongs to exactly one set and the sum of $\Delta_k$'s in every set equals 0. We show that the first step of the algorithm is optimal. That is, there exists an optimal solution that contains exactly the same set of singletons and pairs as 
in the partition returned by the algorithm. Suppose that the optimal solution breaks one pair $\{i,j\}$ ($\Delta_i= - \Delta_j$) and puts $i$ in $S$ and $j$ in $T$. Then we can replace sets $S$ and $T$ with two new sets $\{i,j\}$ and $S \cup T \setminus \{i,j\}$. The new solution has the same cost as before; the sum of $\Delta_k$'s in every set is 0, but the pair $\{i,j\}$ belongs to the partitioning. Repeating this procedure several times, we can transform an arbitrary optimal solution into an optimal solution that contains the same set of singletons and pairs as the solution obtained by the approximation algorithm.

For the sake of the presentation let us assume that $\varepsilon = 0$ (that is, we assume that the approximation algorithms due to Hurkens and Schrijver~\cite{HS}, we use in our algorithm, have approximation guarantees at most
$3/2$ and $2$). Let $p_k$ be the number of sets of size $k$ in the optimal solution. The cost of the optimal solution is $p_2 + 2 p_3 + 3 p_4 + 4p_5 + \cdots$, because the objective function charges $|S|-1$ to a set of size $|S|$. Our approximation algorithm also finds $p_1$ singleton sets and $p_2$ pairs. Then with probability $2/3$, it finds $s_3\geq (2/3) p_3$ triples and covers the remaining $3\cdot (p_3 - s_3) + 4p_4 + 5p_5 + \cdots$ vertices with one set; and with probability $1/3$, it finds $s_4\geq p_4/2$ quadruples and covers the remaining  $3p_3 + 4 \cdot (p_4 - s_4) + 4p_4 + 5p_5+\cdots$ vertices with one set. Thus the expected cost of the solution returned by the algorithm equals
\begin{multline}
\frac{2}{3}\Big(p_2 + 2 \cdot \frac{2 p_3}{3} + 3\cdot \frac{p_3}{3} + 4p_4 + \sum_{k \geq 5}k p_k - 1\Big) + 
\frac{1}{3}\Big(p_2 + 3 \cdot \frac{p_4}{2} + 3 p_3 + 4\cdot \frac{p_4}{2} + \sum_{k \geq 5}k p_k - 1\Big) \\=
p_2 + \frac{23}{9} p_3 + \frac{23}{6}p_4 + \sum_{k \geq 5}k p_k-1.
\end{multline}
Therefore, the approximation ratio of the algorithm, assuming that $\varepsilon = 0$, is 
$$
\frac{p_2 + \frac{23}{9} p_3 + \frac{23}{6}p_4 + \sum_{k \geq 5}k p_k-1}{p_2 + 2p_3 + 3p_4 + \sum_{k \geq 5}(k-1) p_k}
\le \max \left\{\frac{1}{1}, \frac{\frac{23}{9}}{2}, \frac{\frac{23}{6}}{3}, \frac{5}{4}, \frac{6}{5}, \dots\right\} = \frac{23}{18}.
$$
It is easy to verify that if $\varepsilon > 0$, the approximation ratio of the algorithm is at most $23/18 + O(\varepsilon)$. 
\end{proof}

We now prove that finding the exact solution of the problem is NP-hard.

\begin{lemma}
The zero-weight partition problem is NP-hard.
\end{lemma}
\begin{proof}
We construct a reduction from the classical NP-complete {\sc 3-Partition}  to the zero-weight partition problem. Recall that in {\sc 3-Partition} we are given $3m$ numbers $b_1,\ldots, b_{3m}$ strictly between $B/4$ and $B/2$ and we need to check if the set can be partitioned into $m$ sets such that the sum of all elements in each set equals $B$ (and hence each set must have size 3). Such a partition is a ``3 partition.''  
Given an instance of {\sc 3-Partition}, we create $3m$ vertices each having weight $\Delta_k = b_k$. Then we create $m$ vertices each with weight $\Delta_k = -B$. It is easy to see that every set of weight zero must have at least four elements; moreover if the set contains exactly four elements then one of the elements equals $-B$ and 
the other three sum up to $B$. Thus a $3$ partition exists in the original problem if an only if the vertices in the new problem can be partitioned into $m$ zero-weight sets, i.e., the value of the new problem is $4m - m = 3m$.
\end{proof}
\begin{corollary}
One-dimensional {\sc AllRects} is NP-hard.
\end{corollary}

\subsection{The 2-Dimensional Case}
We now consider the 2-dimensional case. We are given an $m \times n$ matrix $A=(a_{ij})$ ($1\leq i \leq m$, $1\leq j \leq n$) and we need to cover it with the minimum number of weighted rectangles $Rect (i_1,i_2,j_1,j_2)$ (for 
arbitrary $i_1,i_2,j_1,j_2$);  
we use  $w(i_1,i_2,j_1,j_2)$ for the weight of $Rect(i_1,i_2,j_1,j_2)$. 
We assume that $a_{ij} = 0$ for $i$ and $j$ outside the rectangle $\{1,\dots, m\}\times
\{1,\dots,n\}$.

By analogy to the 1-dimensional case, define 
$\Delta_{ij} = a_{i,j} - a_{i,j+1} +a_{i+1,j+1}- a_{i+1,j}.$
Call a pair $(i,j)$ with $0\le i\le m$, $0\le j\le n$, with $\Delta_{ij}\ne 0$ an {\em array corner}.
Imagine that the matrix is written in an $m\times n$ table, and $\Delta_{ij}$'s are written at the grid nodes.
The key point is that every rectangle covers exactly one, two, or four of the cells 
$(i+1, j+1)$, $(i,j)$, $(i,j+1)$, $(i+1,j)$ bordering a grid point, and that those covering two or four of those cells cannot affect
$\Delta_{ij}$.
This means that only rectangles having a corner at the intersection of the $i$th and $j$th
grid line contribute to  $\Delta_{ij}$.   (This is why the definition of $\Delta_{ij}$ was ``by analogy'' to the 1-d case.)
In other words,
\begin{eqnarray}\label{eq:Delta2D_Weights}
\Delta_{ij} &=& \sum_{i_2\geq i+1 \mbox{ and }j_2 \geq j + 1} w(i+1,i_2,j+1, j_2) + 
\sum_{i_1\leq i \mbox{ and }j_1 \leq j_2} w(i_1,i,j_1, j)\\
\notag
&\phantom{=}& -\sum_{i_2\geq i+1 \mbox{ and } j_1 \leq j} w(i+1,i_2,j_1,j) -
\sum_{i_1\leq i \mbox{ and }j_2 \geq j+1} w(i_1,i,j+1,j_2).
\end{eqnarray}
This means that the number of rectangles in the optimal solution
must be at least one quarter of the number of array corners, the ``one-quarter"
arising from the fact that each rectangle has exactly four corners and can hence be responsible for at most 
four of the array corners.

It is easy now to give a 4-approximation algorithm, which we sketch without proof, based on this observation.
Build a matrix $M$, initially all zero, which will eventually equal the input matrix $A$.
Until no more array corners exist in $A-M$, find an array 
corner $(i,j)$ with $i<m$ and $j<n$.    (As long as array corners exist, there must be one with $i<m$ and $j<n$.)  
Let $\Delta\ne 0$ be $\Delta_{ij}$.
Add to $M$ a rectangle of weight $\Delta$ with upper left corner 
at $(i,j)$ and extending as far as possible to the right and downward, eliminating the array corner at $(i,j)$ in $A-M$.  

It is easy to see that (1) when the algorithm terminates, $M=A$, and that (2) the number of rectangles used is at most 
the number of array corners in $A$, and hence at most $4|OPT_2(A)|$.

Now we give, instead, a more sophisticated, $23/9+\varepsilon<2.56$-approximation algorithm for the 2D problem.
The idea is to make more efficient use of the rectangles.  Instead of using only one corner of each (in contrast to the adversary,
who might use all four), now we will use two.  In fact, we will deal separately with different horizontal (between-consecutive-row) grid lines, using 
a good 1-dimensional approximation algorithm to decide how to eliminate the array corners on that grid line.  Every time 
the 1-d algorithm tells us to use an interval $[j_1,j_2]$, we will instead inject a rectangle which 
starts in column $j_1$ and ends in column $j_2$,
and extends all the way to the bottom.  Because we use 2 of each rectangle's 4 corners, we pay a price of a factor 
of $4/2$ over the 1-d approximation ratio of
$23/18+O(\varepsilon)$.  Hence we will get $23/9+O(\varepsilon)$.

Here are the details.  Fix $i$ and consider the restriction of the zero-weight partition problem to the $i$th horizontal grid line, i.e., the 1-dimensional zero-weight partition problem with $\Delta_{j} = \Delta_{ij}$. Denote by $OPT^i$ the cost of the optimal solution.
The number of rectangles touching the $i$th horizontal grid line
from above or below is at least $OPT^i$, 
since only these
rectangles contribute $\Delta_{ij}$'s. Every rectangle touches only two horizontal grid lines, thus the total number of rectangles is at least $\sum_{i=1}^m OPT^i/2$.

All rectangles generated by our algorithm will touch the bottom
line of the table;  that is why we lose a factor of 2. Note that
if we could solve the 1-dimensional problem exactly we would be
able to find a covering with $\sum_{i=1}^m OPT^i$ rectangles and
thus get a 2 approximation. For each horizontal grid line $i$,
the algorithm solves the 1-dimensional problem (with $\Delta_{j} = \Delta_{ij}$)
and finds a set of intervals $[j_1,j_2]$ with weights $w_{j_1j_2}$. These intervals are the top sides of the rectangles
generated by the algorithm. All bottom sides of the rectangles lie on the bottom grid line of the table. That is, for every interval $[j_1,j_2]$ the algorithm adds the rectangle $Rect(i,m,j_1,j_2)$ to the solution and sets its weight $w(i,m,j_1,j_2)$ to be $w_{j_1j_2}$. 

The total number of rectangles in the solution output  by the algorithm is 
$\sum_{i=1}^m ALG_i$, where $ALG_i$ is the cost of the solution of the 1-dimensional problem.
 Thus the cost of the solution is at most $2\cdot (23/18 +
O(\varepsilon))$ times
the cost of the optimum solution. We now need to verify that the set of rectangles output by the algorithm is indeed is
a solution.

Subtract the weight of each rectangle from all $a_{ij}$'s covered by the rectangle.
We need to prove that the residual matrix 
$$a'_{ij} = a_{ij} - \sum_{i_1,j_1,j_2: (i,j) \in Rect(i_1,m,j_1,j_2)} w(i_1,m,j_1,j_2)$$
equals zero. Observe that $\Delta'_{ij} = a'_{i+1,j+1} + a'_{ij} - a'_{i+1,j} - a'_{i,j+1} = 0$ for 
all $0 \leq i \leq m-1$ (i.e., all rows $i$, possibly, except for the bottom line) and $0 \leq j \leq n$. Assume that not all $a'_{ij}$ equal to 0. Let
$a'_{i_0j_0}$ be the first nonzero $a'_{ij}$
with respect to the lexicographical order on $(i,j)$. Then $a'_{i_0-1,j_0-1} = a'_{i_0-1,j_0} = a'_{i_0,j_0-1} = 0$. Thus
$a'_{i_0j_0} = 0$. 

We have proven the following theorem.

\begin{theorem}
For every positive $\varepsilon$, there exists a polynomial-time approximation algorithm for {\sc AllRects} 
with approximation guarantee at most $23/9 + O(\varepsilon)$, with $23/9=2.5555....$.
\end{theorem}

\subsection{A Simplified Algorithm}
Because of the dependence on $\varepsilon$, the running time of the previous algorithm can be large
when $\varepsilon$ is small.  A simpler algorithm for the 1-dimensional case---namely, just use pairs and triples---can 
be shown to give ratio $4/3$ for the 1-d case, and hence $8/3=2.6666...$ in 2-d, only slightly worse than $23/9$.  
For the simplified 1-d algorithm, the running time is $O(n+k^2 \log k)$, if there are $k$ $\Delta$'s.  
%since one needs to find triples adding to 0.  One tries all pairs, and then does binary search to find 
% a third value which adds to 0, if possible.  If so, this triple is deleted.
To run the 2-d algorithm,
the running time becomes $O(n^2+\sum_{i=1}^n k_i^2 \log k_i)$, where there are $k_i$ corners on the $i$th row. 
Since the number of corners is $\Theta(OPT)$, the running time is at most $O(n^2)$ plus $O(\max_{k_1+k_2+\cdots+k_n=OPT}
\sum_i k_i^2 \log k_i)$.  Since $f(x)=x^2 \log x$ is convex, this quantity is maximized by making as many $k_i$'s equal to 
$n$ as possible.  A simple proof then shows that the time is $O(n^2+OPT\cdot(n \log n))$.

\section{Acknowledgment}
We thank Divesh Srivastava for initial conversations which inspired
this work.

\appendix
\section{Proof of Theorem~\ref{thm:algmain}}
In the main part of the paper we proved that the expected cost of the solution
returned by the algorithm is at most $4|OPT_2(A)|$. We now  improve this bound to $2|OPT_2(A)|$.

\begin{proof}
We have shown (see bounds~\eqref{eq:TotalBound} and~\eqref{eq:SumU}) that the expected cost of the solution is bounded by
$$\E{\sum_{v\in path(root)} OPT_1 (V_v)}  + 2\sum_{u\neq root}\alpha_u\E{\sum_{v\in path(u)} OPT_1 (V_v) }.$$
Write
$$\E{\sum_{v\in path(root)} OPT_1 (V_v)} + 2\sum_{u\ne root}\alpha_u\E{\sum_{v\in
path(u)} OPT_1 (V_v)} $$
$$= \sum_{v} OPT_1 (V_v) \cdot \left(\Prob{v \in path(root)} + 2\sum_{u\neq root}
\alpha_{u} \Prob{v \in path(u)}\right).$$

Fix a node $v\ne root$. Let $p^0(v)=v$; let $p^1(v)= p(v)$ be the parent of $v$; let $p^2(v) = p(p(v))$ be the grandparent, etc. Finally, let $p^{k}(v)$, say, be the root, $k$ depending implicitly on
$v$. Node $p^0(v) = v$ belongs to $path(v)$ with probability $1$; $v$ belongs to the $path (p^1(v))$ with probability
 $1/d(p^1(v))$; it belongs to $path (p^2(v))$ with probability  $1/(d(p^1(v))d(p^2(v)))$, etc. It belongs to $path(u)$ with probability 0 if $u$ is
not an ancestor of $v$.  Thus
\begin{eqnarray*}
\Prob{v \in path(root)} + 2\sum_{u\neq root} \alpha_{u} \Prob{v \in path(u)} &=& 
%%\Prob{v \in path(root)} + 2\sum_{i=0}^{k-1} \alpha_{p^{i}(v)} \Prob{v \in path(p^i(v))}.
\frac{1}{d(p^1(v))d(p^2(v))\cdots d(p^k(v))} \\&\phantom{=}&+ 2\sum_{i=0}^{k-1} \frac{\alpha_{p^{i}(v)}}{d(p^1(v))d(p^2(v))\cdots d(p^i(v))}
\end{eqnarray*}
Substituting
$$\alpha_{p^{i}(v)} = \frac{d(p(p^i(v))) - 1}{d(p(p^i(v)))} = \frac{d(p^{i+1}(v)) - 1}{d(p^{i+1}(v))},$$
we get a telescoping sum
\begin{align*}
\Pr\Big(&v \in path(root)\Big) + 2\sum_{u\neq root} \alpha_{u} \Prob{v \in path(u)} \\
& = \frac{1}{d(p^1(v))d(p^2(v))\cdots d(p^k(v))} \\
&\phantom{=\dots}+ 2\biggl[\frac{d(p^1(v)) - 1}{d(p^1(v))} \cdot 1+ \frac{d(p^2(v)) - 1}{d(p^2(v))}
  \cdot \frac{1}{d(p^1(v))} + \frac{d(p^3(v)) - 1}{d(p^3(v))} \cdot \frac{1}{d(p^1(v))d(p^2(v))} 
  + \cdots \biggr]\\
&=
\frac{1}{d(p^1(v))d(p^2(v))\cdots d(p^k(v))}\\ 
&\phantom{=\dots} + 2\biggl[\big(1-\frac 1 {d(p^1(v))}\big) +\Big(\frac 1 {d(p^1(v))}-
\frac {1}{d(p^1(v)d(p^2(v))}\Big)\\
&\phantom{=\dots} +\Big(\frac 1 {d(p^1(v))d(p^2(v))}-\frac 1 {d(p^1(v))d(p^2(v))d(p^3(v))}\Big)+\cdots\\
&\phantom{=\dots} +\Big(\frac 1 {d(p^1(v))d(p^2(v))\cdots d(p^{k-1}(v))}-\frac 1
{d(p^1(v))d(p^2(v))\cdots d(p^k(v))}\Big)\biggr]\\
&=2 - \frac{1}{d(p^1(v))d(p^2(v))\cdots d(p^k(v))} < 2.
\end{align*}

Thus
$$\E{\sum_{v\in path(root)} OPT_1 (V_v)}  + 2\sum_{u\neq root}\alpha_u\E{\sum_{v\in path(u)} OPT_1 (V_v) } 
\le \sum_v OPT_1(V_v)\cdot 2.$$

We have proven that the algorithm finds a 2 approximation.  
\end{proof}

\section{NP-hardness of {\sc Tree\texorpdfstring{$\times$}{ x }Tree}}
In this section we sketch a proof that {\sc Tree$\times$Tree} is NP-hard. We show that the problem
is NP-hard even if each of the trees is a star. We construct a reduction from the Directed Hamiltonian Path problem. Let $G=(V,E)$ be a directed graph. Fix a parameter $M = (10 \max\{|V|,|E|\})^4$. For every vertex $v$, we define $M$ rows of our matrix, which we 
denote $R_1(v),\dots, R_M(v)$.
For every directed edge $(u,v)$, we define $M$ columns of our matrix, which we 
denote $C_1(uv),\dots, C_M(uv)$. Thus our matrix has dimensions $(M\cdot |V|) \times (M\cdot|E|)$.
The trees are stars, thus allowed rectangles are the whole matrix, individual rows, individual columns and
individual cells. In our example the gap between the values of ``yes'' and ``no'' instances will be larger than the number of rows plus the number of columns. Thus, we may assume that rectangles corresponding to columns and rows are free to use. In this case, we may also assume that 
the weight of the rectangle covering the whole matrix is 0 (instead of having this rectangle with weight $w$ in the solution we may just increase the value of all columns by $w$). Denote by $x_i(z)$ the variable for the rectangle corresponding to row $R_i(z)$ (possibly 0); denote by $y_j(uv)$ the variable
for the rectangle corresponding to column $C_j(uv)$; denote the entry of the matrix at the intersection of the row $R_i(z)$ and the column  $C_j(uv)$ by
$a_{ij}(z,uv)$. Then the cost of the solution equals the number of individual cells with nonzero weight, i.e., the number of unsatisfied equations 
$$x_i(z) + y_{j} (uv) = a_{ij}(z,uv).$$
Thus the problem is to find values of variables $x_i(z)$ and $y_{j} (uv)$ so as to minimize 
the number of unsatisfied equations. Remember, however, that we need to guarantee a gap of at least 
$M\cdot |V| + M\cdot|E|$ between the values of ``yes'' and ``no'' instances.

We set $a_{ij}(u, uv) = 0$ for every vertex $u$ and  every edge $(u,v)$. We set $a_{ij}(v, uv) = 1$ for every vertex $v$ and every edge $(u,v)$.  We call the rest of the matrix entries, i.e., entries $a_{ij}(z,uv)$, where $z\ne u$ and $z\ne v$, ``bad entries.'' Let us pretend for a while that there are no bad entries and that there are no equations corresponding to bad entries.
(Later we will set $a_{ij}(z,uv)=ij$.)

We claim that if the graph has a directed Hamiltonian path then there exists a solution 
with at most $(|E| - |V| + 1) \cdot M^2$ unsatisfied equations. Let $pos(u)$ be the position of the vertex in
the Hamiltonian path: 1st, 2nd, 3rd, etc. Then we set $x_i(u) = pos (u)$ and
$y_j(uv) = -pos(u)$. Observe that if an edge $(u,v)$ belongs to the Hamiltonian path,
then 
$$x_i(u) + y_j(uv) = pos(u) - pos(u) = 0 = a_{ij}(u,uv)$$
and 
$$x_i(v) + y_j(uv) = (pos(u) + 1) - pos(u) = 1 = a_{ij}(v,uv).$$
If an edge $(u,v)$ does not belong to the Hamiltonian path,
then still
$$x_i(u) + y_j(uv) = pos(u) - pos(u) = 0 = a_{ij}(u,uv),$$
but
$$x_i(v) + y_j(uv) = pos(v) - pos(u) \neq 1 = a_{ij}(v,uv).$$
The number of unsatisfied equations thus equals 
$M^2 \cdot (|E|- |V| + 1)$.

Now we show that if the graph does not have a directed Hamiltonian path, then
every solution has cost at least $M^2 \cdot (|E|- |V| + 2)$. Assume to the contrary,
that there exists a solution of cost less than $M^2 \cdot (|E|- |V| + 2)$.
Since all variables $x_i(u)$ for a fixed $u$ and $i=1,\dots, M$  participate in exactly the same equations
we may assume that $x_i(u) = x_j(u)$ for all $i$ and $j$ in the optimal solution. Similarly,
we may assume that $y_{i}(uv) = y_{j}(uv)$ for all $i$ and $j$. (Recall that we now ignore all
bad equations.) If $x_i(z) + y_j(uv) = a_{ij}(z,uv)$ ($z=u$ or $z=v$), then
the same equality holds for every $i$ and $j$. Thus, the number of unsatisfied equations
is at most $M^2 \cdot (|E|- |V| + 1)$ (since the number of unsatisfied equations is divisible
by $M^2$). 
Consider an edge $(u,v)$ for which 
$x_i(u) + y_j(uv) = 0$ and $x_i(v) + y_j(uv) = 1$. We have $x_i(v) - x_i(u) = 1$. 
The number of such edges is at least $|V|-1$ (since the number 
of edges for which  $x_i(u) + y_j(uv) \neq 0$ or  $x_i(v) + y_j(uv) \neq 1$ is at most
the total number of unsatisfied equations divided by $M^2$, i.e., $|E|- |V| + 1$, and 
the total number of edges is $|E|$). Therefore, if we place vertex $u$ at position 
$x_i(u)+(1-\min_s x_i(s))$ (recall that $x_i(u)$ does not depend on $i$) we get a Hamiltonian path.

We are almost done. We only need to take care of bad equations. The idea is to set the rest of values
$a_{ij}(z,uv)$ so that only very few bad equations can be satisfied. For each $z$ and every edge $(u,v)$ we define an $M \times M$ matrix  $a_{ij}(z,uv) = ij$. We claim that in every matrix 
$a_{ij}(\cdot, \cdot)$, the number of satisfied equations is at most $3M^{3/2}$. We prove the claim in Lemma~\ref{lem:numbadentries}.
Then for every assignment of variables $x_i(z)$ and $y_{j}(uv)$, the total number of satisfied bad equations is at most $|E|\cdot|V|\cdot 3M^{3/2} < M^2/2$. Hence, the gap between ``yes" and ``no"
instances is at least $M^2/2$.

\begin{lemma}\label{lem:M32}
Consider an $M\times M$ matrix $a_{ij}$ of zeros and ones. Suppose that for every $i_1$, $i_2$, $j_1$ and $j_2$ ($i_1\ne i_2$ and $j_1 \ne j_2$) at most three out of four of values $a_{i_1j_1}$,
$a_{i_1j_2}$, $a_{i_2j_1}$, $a_{i_2j_2}$ equal 1. Then the number of ones in the matrix is 
at most $3M^{3/2}$.
\end{lemma}
\begin{proof}
Perform the following algorithm:
While there exists a column containing at least $\sqrt{M}$ ones, pick one such column $j$. Remove all rows $i$ of the at-least-$\sqrt M$ rows
that have 1 at the intersection with column $j$.

When the algorithm stops, the remaining matrix has at most $M^{3/2}$ ones. Let $R_t\geq \sqrt{M}$ be the number of rows removed at step $t$. At every step $t$, the algorithm removes $MR_t$ entries, among which there are at
most $R_t+(M-1)$ ones ($R_t$ ones in the selected column
and at most one in each of the remaining $M-1$ columns, by
hypothesis). Hence, the fraction of removed ones among all removed entries is at most $(R_t+M)/(MR_t)=1/M+1/R_t$. Thus the total number of removed ones is at most $M^2(1/M+1/R_t)\le
M+M^{3/2}$. We get that the total number of ones present in the original matrix is at most $M+M^{3/2}$ plus the at-most-$M^{3/2}$ ones in the
resulting matrix, or at most $M+2M^{3/2}$ in total. 
\end{proof}
\begin{lemma}\label{lem:numbadentries}
Consider a system of linear equations 
$$x_i + y_j = {ij}.$$
For all possible $x_i$ and $y_j$ the number of satisfied equations is at most $3M^{3/2}$.
\end{lemma}
\begin{proof}
Observe that for every $i_1$, $i_2$, $j_1$ and $j_1$ ($i_1\ne i_2$ and $j_1 \ne j_2$), 
it is not possible to satisfy all four equations: $x_{i_1} + y_{j_1} = {i_1j_1}$,
$x_{i_1} + y_{j_2} = {i_1j_2}$, $x_{i_2} + y_{j_1} = {i_2j_1}$, and $x_{i_2} + y_{j_2} = {i_2j_2}$, since if all four of them are satisfied then 
$${i_1j_1} + {i_2j_2} = x_{i_1} + y_{j_1} + x_{i_2} + y_{j_2} = {i_1j_2} + {i_2j_1},$$
but 
$i_1j_1 + i_2j_2 \neq i_1j_2 + i_2j_1$
(since $i_1(j_2 - j_1) \neq i_2(j_2 - j_1)$).
Lemma~\ref{lem:M32} now implies that the number of satisfied equations is at most $3M^{3/2}$.  
\end{proof}
\section{A Running Time Comparison Between The Present Algorithms And Natarajan's} \label{comparison}
Of course it is not fair to compare our algorithms, which
approximately solve the exact problems,
with Natarajan's, which approximately solves the inexact $L_2$ problem.  Of course
the optimal value for our problem, being exact, is at least as large as the optimal value for
Natarajan's problem.  While Natarajan's algorithm is very general, the price paid is that it's slow.

For problem {\sc Tree$\times$Tree}, our algorithm takes time $O(dn^2)$ 
in total, which is $O(d)$ times the input size of $n^2$,
where $d<n$ is the smaller of the depths of the two trees;  typically one expects $d$ to be $O(\log n)$ 
(or constant) in applications.  Natarajan's algorithm takes time
$\Omega(n^4)$ even for each iteration.  
% because each iteration involves a matrix-vector product with
% the matrix being $n^2 \times O(n^2)$.

For problem {\sc AllRects}, the contrast between the running times of our algorithm
and Natarajan's is even more stark.  Our simplified
$8/3$-approximation algorithm runs in time $O(n^2+OPT\cdot(n\log n))$
(where the input size is $n^2$)
with $OPT\le n^2$,
whereas Natarajan's
takes time $\Omega(n^6)$ per iteration.  This makes Natarajan's algorithm wildly impractical 
for the large instances which often occur in database applications.  

\begin{thebibliography} {MM}
\bibitem{ABGYKS}
D. Agarwal, D. Barman, D. Gunopulos, N. Young, F. Korn, and D.
Srivastava, ``Efficient and Effective Explanation of Change In 
Hierarchical Summaries," Proc. {\em ACM SIGKDD International
Conference on Knowledge Discovery and Data Mining (KDD)}, 2007, 6-15.
\bibitem {ACJKLW}
D. Applegate, G. Calinescu, D. S. Johnson, H. Karloff, K.
Ligett, and J. Wang.
``Compressing Rectilinear Pictures and Minimizing Access Control
Lists," 
SODA 2007, 1066-1075.
\bibitem {AKR}
V. S. Anil Kumar and H. Ramesh, ``Covering Rectilinear Polygons
With Axis-Parallel Rectangles,'' STOC 1999, New York, 445-454.
\bibitem {BCS}
N. Bansal, D. Coppersmith, and B. Schieber. ``Minimizing Setup
and Beam-On Times in Radiation Therapy,'' APPROX 2006, 27-38.
\bibitem{Bu05}
S.~Bu, V. S. Lakshmanan, and R. T. Ng.
\newblock {MDL Summarization with Holes}.
\newblock In {\em VLDB '05: Proceedings of the 31st international conference on
  Very Large Databases}, pages 433--444, VLDB Endowment, 2005.
\bibitem{CT}
\newblock E. Candes and T. Tao. {Decoding By Linear Programming}. 
\newblock In {\em IEEE Transactions on
Information Theory} 51 (12), 2005, 4203-4215.
\bibitem{Cormode04}
G.~Cormode, F.~Korn, S.~Muthukrishnan, and D.~Srivastava.
\newblock {Diamond in the Rough: Finding Hierarchical Heavy Hitters in
  Multi-Dimensional Data}.
\newblock In {\em Proc. of ACM SIGMOD '04}, Paris, France, 2004.
\bibitem{Donoho}
D. Donoho.
\newblock {For Most Large Underdetermined Systems of Linear Equations the Minimal
$\ell^1$-norm Solution is also the Sparsest Solution}, 
\newblock In {\em Communications on Pure and Applied Mathematics} 59 (6), June 2006,
797-829.
\bibitem{FK} A. Frieze and R. Kannan.
\newblock {Quick Approximation to Matrices and Applications}.
\newblock In {\em Combinatorica} 19 (2), 175-220, 1999.
\bibitem{HS} C. Hurkens and A. Schrijver.
``On the Size of Systems of Sets Every $t$ of Which Have an SDR, With an Application to the Worst-Case Ratio of Heuristics for Packing Problems,''
{\em SIAM J. Discrete Math.} 2(1), 1989, 68-72.
\bibitem{Karpinski09}
M.~Karpinski and W.~Schudy.
\newblock {Linear Time Approximation Schemes for the Gale-Berlekamp Game and
  Related Minimization Problems}.
\newblock In {\em STOC '09: Proceedings of the 41st Annual ACM Symposium on
  Theory of Computing}, pages 313--322, New York, NY, USA, 2009. ACM.
\bibitem{Lakshmanan02}
V.S. Lakshmanan, R.T. Ng, C.~Xing Wang, X.~Zhou, and T.~Johnson.
\newblock {The Generalized MDL Approach for Summarization}.
\newblock In {\em VLDB}, pages 766--777, 2002.
\bibitem {Natarajan}
B. K. Natarajan. ``Sparse Approximate Solutions To Linear
Systems," {\em SIAM J. Comp.} 24(2), 1995, 227-234.
\bibitem{Roth08}
R.~M. Roth and K.~Viswanathan.
\newblock {On the Hardness of Decoding the Gale-Berlekamp Code}.
\newblock {\em IEEE Transactions on Information Theory}, 54(3):1050--1060,
  2008.
\bibitem{Sarawagi99}
S.~Sarawagi.
\newblock {Explaining Differences in Multidimensional Aggregates}.
\newblock In {\em Proc. of the 25th International Conference on Very Large
  Databases (VLDB)}, pages 42--53, Scotland, UK, 1999.
\bibitem{W}
Wikipedia page {\tt http://en.wikipedia.org/wiki/Restricted\_isometry\_property}.
\end{thebibliography}
\end{document}